%% file: expectation.tex
\documentclass[11pt]{article}
\usepackage{geometry}
\usepackage{amssymb}
\usepackage{mathrsfs}
\usepackage{graphicx}
\usepackage{amsmath}

\usepackage{epsfig}

\newtheorem{theorem}{Theorem}[section]

\newtheorem{lemma}[theorem]{Lemma}
\newtheorem{proposition}[theorem]{Proposition}
\newtheorem{claim}[theorem]{Claim}
\newtheorem{definition}[theorem]{Definition}

\def\squarebox#1{\hbox to #1{\hfill\vbox to #1{\vfill}}}
\newcommand{\qed}{\hspace*{\fill}
\vbox{\hrule\hbox{\vrule\squarebox{.667em}\vrule}\hrule}\smallskip}

\newenvironment{proof}{\noindent{\bf Proof:~~}}{\(\qed\)}

\newcommand{\ignore}[1]{}
\newcommand{\vnoin}{\vspace{0.1in}\noindent}

\newcommand{\set}[1]{\{ #1 \}}

\newcommand{\floor}[1]{\lfloor {#1} \rfloor}

\newcommand{\lt}{\left}
\newcommand{\rt}{\right}

\DeclareMathOperator{\poly}{poly}

\def\min{\qopname\relax n{min}}
\def\max{\qopname\relax n{max}}

\def\argmax{\qopname\relax n{argmax}}

\def\Pr{\qopname\relax n{\mathbf{Pr}}}
\def\Ex{\qopname\relax n{\mathbf{E}}}

\def\R{\mathcal{R}}
\def\A{\mathcal{A}}

\newcommand{\comment}[1]{}

\DeclareMathOperator{\weight}{weight}
\DeclareMathOperator{\rweight}{r-weight}

\begin{document}

\title{On the Power of Randomization in Algorithmic Mechanism Design}\date{}
\author{Shahar Dobzinski\thanks{Supported by the Adams Fellowship
    Program of the Israel Academy of Sciences and Humanities, and by a
    grant from the Israeli Academy of
    Sciences.}\\
  School of Computer Science and Engineering\\
  Hebrew University of Jerusalem\\
  {\tt shahard@cs.huji.ac.il}
\and Shaddin Dughmi \thanks{Supported in part by NSF grant CCF-0448664.}\\
  Department of Computer Science\\
  Stanford University\\
  {\tt shaddin@cs.stanford.edu}
  }

\maketitle \setcounter{page}{0} \thispagestyle{empty}

\begin{abstract}
In many settings the power of truthful mechanisms is severely bounded. In this paper we use randomization to overcome this problem. In particular, we construct an FPTAS for multi-unit auctions that is truthful in expectation, whereas there is evidence that no polynomial-time truthful deterministic mechanism provides an approximation ratio better than $2$.

We also show for the first time that truthful in expectation polynomial-time mechanisms are \emph{provably} stronger than polynomial-time universally truthful mechanisms. Specifically, we show that there is a setting in which: (1) there is a non-polynomial time truthful mechanism that always outputs the optimal solution, and that (2) no universally truthful randomized mechanism can provide an approximation ratio better than $2$ in polynomial time, but (3) an FPTAS that is truthful in expectation exists.
\end{abstract}

\newpage

\input{intro}
\input{prelim}
\input{constant}
\input{general}
\input{separation}

\subsubsection*{Acknowledgements}
We thank Bobby Kleinberg, Noam Nisan, Sigal Oren, and Chaitanya Swamy for helpful discussions and comments.

\bibliographystyle{plain}
\bibliography{bib}

\end{document}

%% file: intro.tex
\section{Introduction}

\subsubsection*{Background and Our Results}

The last few years have been quite disappointing for researchers in Algorithmic Mechanism Design. Indeed, much progress has been made, but in the ``wrong'' direction: several papers proved that the power of polynomial time truthful mechanisms is severely bounded, either comparing to polynomial time non-truthful algorithms, or to non-polynomial time truthful mechanisms \cite{LMN03,DN07a,DN07b,DS08-add,PSS08,MPSS09}. This paper brings some good news: randomization might help in breaking these lower bounds to achieve better approximation guarantees.

Before studying the power of randomization in the context
of mechanism design, we recall the three common notions of
truthfulness:

\begin{itemize}
\item \textbf{Deterministic Truthfulness: }A bidder always maximizes
  his utility by bidding truthfully. No randomization is allowed.
\item \textbf{Universal Truthfulness: } A universally-truthful mechanism is a probability distribution over deterministic truthful mechanisms. The mechanism is truthful even when the realization of the random coins is known.
\item \textbf{Truthfulness in Expectation: }A mechanism is truthful in expectation if a bidder always maximizes his \emph{expected} profit by bidding truthfully. The expectation is taken over the internal random coins of the mechanism.
\end{itemize}

It is well known that, in some settings (e.g. online), randomized
algorithms are strictly more powerful than determinic
algorithms. In their seminal paper introducing algorithmic mechanism design, Nisan and Ronen \cite{NR99} showed that universally truthful mechanisms can be more powerful than their deterministic counterparts. In this paper we show that truthfulness in expectation yields more power than universal truthfulness.

For the well-studied problem of multi-unit auctions (e.g., \cite{MN02,BKV05,KPS04,NS06,LS05,DN07b,BBM08}), we show a polynomial time truthful in expectation mechanism with the best ratio possible from a pure algorithmic point of view:

\vnoin \textbf{Theorem: }There exists a truthful in expectation FPTAS for multi-unit auctions\footnote{An FPTAS is a $(1+\epsilon)$-approximation algorithm with running time that is also polynomial in $\frac 1 \epsilon$.}.

\vnoin There are many truthful in expectation mechanisms in the
literature \cite{AT01,APTT03,LS05,DDDR08}. However, somewhat
surprisingly, many truthful in expectation mechanisms were followed by
universally truthful mechanisms with the same performance
\cite{BKV05,D07,DNS06,DN07b}. As a result, there are few settings in
which the best known truthful in expectation algorithms provide better
approximation ratios than the best known universally truthful
mechanisms. Furthermore, \cite{MV04} shows that for digital good
auctions the two notions of randomization are equivalent in power. In
contrast, our second main result shows that truthful in expectation
mechanisms are strictly more powerful than universally truthful
ones. We study a variant of multi-unit auctions that we term
\emph{restricted multi-unit auctions}, and show a first-of-a-kind
separation result:

\vnoin \textbf{Theorem: } If $\A$ is a universally truthful mechanism for
restricted multi-unit auctions that achieves an approximation ratio of
$2-\epsilon$ for some constant $\epsilon>0$, then $\A$ has exponential
communication complexity. However, there exists a truthful in expectation FPTAS for restricted multi-unit auctions.

\vnoin We note that there exists a deterministic truthful mechanism that optimally solves restricted multi-unit auctions in exponential time.

\subsubsection*{Multi-unit Auctions}

In a multi-unit auction a set of $m$ identical items is to be
allocated to $n$ bidders. Each bidder $i$ has a valuation function
$v_i:[m]\rightarrow \mathbb {R^+}$, where $v_i$ is non-decreasing, and
normalized: $v_i(0)=0$. The goal is the usual one of finding an
allocation of the items $(s_1,\ldots,s_n)$ that maximizes the social
welfare: $\Sigma_i v_i(s_i)$. All items are identical, so algorithms
should run in time polynomial in the number of bits needed to
represent the number $m$, and the number of bidders: $\poly(n,\log
m)$.

It is not hard to see that in the ``single minded'' case, where each $v_i$ is a single-step function (i.e., for some $s^*,k>0$, $v_i(s)=k$ if $s\geq s^*$ and $v_i(s)=0$ otherwise), the problem is just a reformulation of the NP-hard Knapsack problem. The standard FPTAS for knapsack generalizes easily to multi-unit auctions.

The study of truthfulness in multi-unit auctions has a long history,
starting with Vickrey's 1961 paper \cite{Vic61}. The VCG mechanism is
truthful and solves the problem optimally, but is not
computationally-efficient (see, e.g., \cite{Nis07}). For the
single-minded case, Mu'alem and Nisan provided a truthful
polynomial-time $2$-approximation algorithm, followed by an FPTAS by
Briest et al \cite{BKV05}. In the general case, a truthful in
expectation $2$-approximation mechanism was presented by \cite{LS05},
followed by a deterministic $2$-approximation mechanism in
\cite{DN07b}. Can polynomial-time truthful mechanisms guarantee an
approximation ratio better than $2$? This is one of the major open
questions in algorithmic mechanism design, and only a partial answer
is known: a deterministic truthful mechanism with an approximation
ratio better than $2$ that \emph{always allocates all items} must run
in exponential time\footnote{Indeed, it can be assumed that a
  non-truthful algorithm always allocates all items without loss of
  generality from a pure algorithmic perspective. However, this
  assumption has game-theoretic implications.} \cite{LMN03,DN07b}. As
mentioned before, this paper provides an FPTAS for multi-unit auctions
that is truthful in expectation.

A word is in order on how the valuations are accessed. The valuation functions are objects of size $m$, whereas we are interested in algorithms that run in time polynomial in $\log m$ (the running time of the non-truthful FPTAS for multi-unit auctions). Hence, we assume that each valuation $v$ is given by a black box. For our upper bounds, the black box corresponding to $v$ needs to answer only the weak ``value queries'': given $s$, what is the value of $v(s)$. Our lower bound for the power of universally truthful mechanisms assumes a black box that can answer any query based on $v$ (the ``communication model'').

\subsubsection*{Main Result I: A Truthful FPTAS for Multi-unit Auctions}

We now give a short description of our truthful FPTAS for multi-unit auctions. The only technique known for designing truthful deterministic mechanisms for ``rich'' problems like multi-unit auctions is by designing maximal-in-range (henceforth MIR) algorithms. An algorithm $A$ is called \emph{maximal-in-range} if there is a set of allocations $\mathcal R$ (the ``range'') that does not depend on the input, such that $A$ always outputs the allocation in $\mathcal R$ that maximizes the welfare: $A(v_1,...\ldots,v_n)=\argmax_{(s1,\ldots,s_n)\in\mathcal R}\Sigma_i v_i(s_i)$. Using the VCG payment scheme together with an MIR mechanism results in a truthful mechanism. See \cite{DN07a} for a more formal discussion.

Therefore, one way to obtain truthful mechanisms is to identify a range that, on one hand, is ``rich'' enough to provide a good approximation ratio, and, on the other hand, is ``simple'' enough so that exact optimization over this range is computationally feasible. For multi-unit auctions, there is an MIR mechanism that provides an approximation ratio of $2$ in polynomial time, but unfortunately no MIR algorithm can provide a better ratio in polynomial time \cite{DN07b}.

In this paper we let the range $\mathcal R$ consist of \emph{distributions} over allocations. An algorithm that always selects the allocation that maximizes the \emph{expected} welfare and uses VCG payments is truthful in expectation. We term these mechanisms \emph{maximal in distributional range} (MIDR).

Before discussing our truthful FPTAS, recall the spirit of the standard non-truthful FPTAS for multi-unit auctions: each valuation is simplified by rounding down each value to the nearest power of $(1+\epsilon)$. Then, the best solution that uses only the "breakpoints" of the rounded valuations is found. This solution has a value close to the optimal unrestricted solution. However, to guarantee truthfulness via a maximal-in-range algorithm we must find the solution with the \emph{optimal} welfare in the range. Thus, we give a ``weight'' $w_{\vec s}$ to each allocation $\vec s$: when $\vec s$ is selected by the algorithm, the output will be $\vec s$ with probability $w_{\vec s}$, and with probability $1-w_{\vec s}$ no bidder will receive any items, effectively reducing the expected welfare of the allocation by a factor of $w_{\vec s}$. The weight of a less ``structured'' allocation is smaller than the weight of a more ``structured'' one. We show that in the optimal solution among these ``weighted allocations'' each bidder is assigned a bundle that is a breakpoint, or very close to a breakpoint of his valuation. Thus the optimal solution can be found efficiently using exhaustive search, when the number of bidders is a constant, since the number of breakpoints of each bidder is polynomial in the number of bits.

Efficiently handling any number of bidders is more involved. Given two bidders we define a ``meta-bidder'' by merging their valuations: the value of the meta bidder for $s$ items is equal to the value of the optimal solution that allocates $s$ items between the two bidders, using the weighted allocations mentioned above. We recursively define new meta-bidders given the previous ones, until we are left with only two meta bidders. Now the optimal solution can be found efficiently.

We note that MIDR mechanisms were used in \cite{LS05}, although only implicitly. The beautiful construction of \cite{LS05} is general and applies to many settings. However, its strength and weakness is its generality: for some settings, a specifically-tailored mechanism might have more power than a mechanism obtained from the general construction of \cite{LS05}. Indeed, for some of the most important settings discussed in \cite{LS05} better constructions have been found \cite{DNS06,D07,DN07b}. Moreover, the LP-based techniques of \cite{LS05} cannot guarantee an approximation ratio better than the integrality gap for multi-unit auctions, which is $2$ (see \cite{LS05}). The techniques in this paper may help in designing better mechanisms in settings where \cite{LS05} performs poorly, like combinatorial auctions with submodular bidders (see below).

\subsubsection*{Main Result II: Truthful in Expectation Mechanisms are more Powerful}

Next we prove that truthful in expectation mechanisms are strictly
more powerful than universally truthful ones. Ideally, we would like
to prove that universally truthful mechanisms for multi-unit auctions
cannot provide an approximation ratio better than $2$ in polynomial
time. However, this remains open question even for deterministic
mechanisms. Hence, we study a close variant of multi-unit auctions,
called \emph{restricted multi-unit auctions}. The FPTAS for multi-unit
auctions extends almost immediately to restricted multi-unit
auctions. It is more involved to show that universally truthful
polynomial time mechanisms cannot provide an approximation ratio
better than $2$.

Roughly speaking, we first show that a polynomial-time universally
truthful mechanism with an approximation ratio of $\alpha$ must yield
a polynomial-time \emph{deterministic} truthful mechanism with an
approximation ratio of $\alpha$ on ``many'' instances. We then show
that this deterministic mechanism must be an affine maximizer (a
slight generalization of MIR algorithms). Finally, we prove that
deterministic affine maximizers cannot provide an approximation
ratio better than $2$ for restricted multi-unit auctions in polynomial
time for ``many'' instances, hence no polynomial-time universally
truthful mechanism for restricted multi-unit auctions with an
approximation ratio better than $2$ exists.

%
%
%

\subsubsection*{Open Questions}

The obvious open question is whether the techniques of this paper can be extended to improve the approximation guarantees of other questions in algorithmic mechanism design, such as variants of combinatorial auctions, or combinatorial public projects \cite{PSS08}. In particular:

\vnoin \textbf{Main Open Question: }Is there a truthful mechanism for combinatorial auctions with subadditive (or submodular) bidders that provides a constant approximation ratio?

\vnoin There are non-truthful algorithms with constant approximation
ratios for combinatorial auctions with subadditive
bidders\cite{F06,V08}. However, the best known polynomial-time
truthful mechanisms provide a ratio that is no better than logarithmic
\cite{D07,DNS06}. It seems that our current techniques reached a dead
end: random sampling methods do not seem amenable to an approximation
factor that is better than logarithmic, deterministic MIR mechanisms
cannot provide an approximation ratio better than $m^{\frac 1 6}$ in
polynomial time \cite{DN07a}, and the techniques of \cite{LS05} are
not applicable to this setting. The methods of this paper might help
in constructing truthful in expectation mechanisms. Alternatively,
\emph{lower bounds} on the power of MIDR mechanisms would also be
extremely interesting.


%% file: prelim.tex
\section{Preliminaries}

\subsection{The Setting}

In a multi-unit auction there is a set of $m$ identical items, and a
set $N=\set{1,2,\ldots,n}$ of bidders. Each bidder $i$ has a valuation
function $v_i:[m]\rightarrow \mathbb R^+$, which is normalized
($v_i(0)=0$) and non-decreasing. Denote by $V$ the set of possible
valuations. An allocation of the items $\vec s=(s_1,\ldots,s_n)$ to
$N$ is a vector of non-negative integers with $\Sigma_is_i\leq
m$. Denote the set of allocations by $S$. The goal is to find an
allocation that maximizes the welfare: $\Sigma_iv_i(s_i)$.

The valuations are given to us as black boxes. For algorithms, the black box $v$ will only answer the weak value queries: given $s$, what is the value of $v(s)$. For the impossibility result, we assume that the black box $v$ can answer any query that is based on $v$ (the ``communication model''). Our algorithms run in time $\poly(n,\log m)$, while our impossibility result gives a lower bound on the number of bits transferred, and holds even if the mechanism is computationally unbounded.

\subsection{Truthfulness}

An $n$-bidder mechanism for multi-unit auctions is a pair $(f,p)$ where $f:V^n \rightarrow S$ and $p=(p_1,\cdots, p_n)$, where $p_i:V^n\rightarrow \mathbb R$. $(f,p)$ might be either randomized or deterministic.

\begin{definition}
Let $(f,p)$ be a deterministic mechanism. $(f,p)$ is \emph{truthful} if for all $i$, all $v_i, v'_i$ and all $v_{-i}$ we have that $v_i(f(v_i,v_{-i})_i)-p_i(v_i,v_{-i})\geq v'_i(f(v'_i,v_{-i})_i)-p(v'_i,v_{-i})$.
\end{definition}

\begin{definition}
$(f,p)$ is \emph{universally truthful} if it is a probability distribution over truthful deterministic mechanisms.
\end{definition}

\begin{definition}
  $(f,p)$ is \emph{truthful in expectation} if for all $i$, all $v_i,
  v'_i$ and all $v_{-i}$ we have that
  $\Ex[v_i(f(v_i,v_{-i})_i)-p(v_i,v_{-i})]\geq
  \Ex[v'_i(f(v'_i,v_{-i})_i)-p_i(v'_i,v_{-i})]$, where the expectation
  is over the internal random coins of the algorithm.
\end{definition}

\subsection{Maximal in Range, Maximal in Distributional range, and Affine Maximizers}

\begin{definition}
$f$ is an \emph{affine maximizer} if there exist a set of allocations $\mathcal R$, a constant $\alpha_i \ge 0$ for $i\in N$, and a constant $\beta_{\vec s} \in \Re $ for each $\vec s\in S$, such that $f(v_1,...,v_n) \in \argmax_{\vec s=(s_1,\ldots,s_n)\in \mathcal R} (\Sigma_i (\alpha_i v_i(s_i)) + \beta_s)$. $f$ is called \emph{maximal-in-range} (MIR) if $\alpha_i=1$ for $i\in N$, and $\beta_s=0$ for each $\vec s\in \mathcal R$.
\end{definition}

\begin{definition}
  $f$ is a \emph{distributional affine maximizer} if there exist a set
  of distributions over allocations $\mathcal D$, a constant $\alpha_i
  \ge 0$ for $i\in N$, and a constant $\beta_{D} \in \Re $ for each
  $D\in \mathcal D$) such that:
  $f(v_1,...,v_n) \in \argmax_{D\in\mathcal D}\lt(
    \Ex_{\vec{s} \sim D} \left[
    \Sigma_i \alpha_i \cdot v_i(s_i)\right] + \beta_D\rt) $.
We say $f$ is \emph{maximal in distributional range} (MIDR) if
$\alpha_i=1$ for all $i$ in $N$, and $\beta_D=0$ for each $D\in
\mathcal D$.
\end{definition}

The following proposition is standard:

\begin{proposition} The following statements are true:

\begin{enumerate}

\item Let $f$ be an affine maximizer (in particular, maximal in range). There are payments $p$ such that $(f,p)$ is a truthful mechanism.

\item Let $f$ be a distributional affine maximizer  (in particular, maximal in distributional range). There are payments $p$ such that $(f,p)$ is a truthful in expectation mechanism.
\end{enumerate}

\end{proposition} 


%% file: constant.tex
\section{A Truthful FPTAS for Multi-Unit Auctions}\label{sec:fptas}

This section provides a $(1+\epsilon)$-approximation algorithm for multi-unit auctions that is truthful in expectation by providing a maximal-in-distributional-range mechanism. The running time is polynomial in $n,\log m$, and $\frac 1 \epsilon$. The construction has several stages. We start by defining a certain family of ranges, structured ranges. The main property of this family is that finding the optimal solution in a structured range is computationally feasible when the number of bidders is fixed. To handle any number of bidders, we use a more involved tree-like construction, where structured ranges serve as the basic building block. We start by defining structured ranges.

\subsection{Structured Ranges}

\subsubsection{Definition}

Towards constructing an FPTAS, fix the approximation parameter $\epsilon$ such that $ 0 < \epsilon < \frac 1 2$. It will be useful to let $\delta=\frac {\ln \frac 1 {1-\epsilon}}{2\log m+2}$. We construct ranges that consist only of weighted allocations:

\begin{definition}\label{def-weight}
A \emph{weighted allocation} $w\cdot \vec s$ is a distribution over allocations of the following form: with some probability $w$ the allocation $\vec s=(s_1,\dots,s_n)$ is chosen, and with probability $1-w$ the empty allocation (where each bidder is allocated an empty set) is chosen. $w$ is called the \emph{weight} of the weighted allocation.
\end{definition}

The weight of a weighted allocation of $m$ items to $n$ bidders is
determined using the following function $\weight_\epsilon:S\rightarrow
[1-\epsilon,1]$, where $S$ is the set of all allocations of $m$ items
to $n$ bidders. Let $t$ be the maximal (non-negative) integer such
that for each $i$, $s_i$ is a multiple of $2^t$. Let $p=(1+2\delta)$
if $s_i\neq 0$ for $n-1$ bidders, let $p=1$ otherwise. Now define
$\weight_\epsilon(\vec s)=(1-\epsilon)\lt(1+2 \delta
\rt)^{t}p$. Observe that the number of different outputs of
$\weight_\epsilon$ is $O(\log m)$, where increasing $t$ by $1$ makes
the output value increase by a factor of $(1+2\delta)$.

A structured range consists only of weighted allocations whose weight is determined by $\weight_\epsilon$. Informally, in a structured range ``structurally simpler'' allocations have greater weight.

\begin{definition}
An \emph{$\epsilon$-structured range} for $n$ bidders and $m$ items is the following set of weighted allocations: $\set{\weight_{\epsilon}(\vec s)\cdot \vec s}_{\vec s\in S}$, where $S$ is the set of allocations for these $n$ bidders. We sometimes call an $\epsilon$-structured range a structured range when $\epsilon$ is clear from context.
\end{definition}

The optimal solution in a structured range provides a good
approximation to the (unweighted) optimal solution: every allocation
is in the range with weight at least $(1-\epsilon)$. In particular,
the optimal (unweighted) allocation is in the range with a weight of at least $(1-\epsilon)$.

\begin{lemma}\label{lemma-approx}
The expected approximation ratio of an algorithm that optimizes over an  $\epsilon$-structured range is at least $(1-\epsilon)$.
\end{lemma}

\subsubsection{The Optimal Solution is in Neighborhoods of Breakpoints}

The important property of structured ranges is that finding an optimal weighted allocation is computationally easy. We show that the optimal solution consists of bundles that are near breakpoints, where breakpoints are bundles at which the valuation first exceeds a power of $1+\delta$.

\begin{definition}
A bundle $s$ is called a \emph{$b$-breakpoint} of a valuation $v$ if it is the smallest multiple of $b$ s.t. $v(s)\geq (1+\delta)^lv(b)$, for some integer $l\geq 0$. The bundle $0$ is also considered a $b$-breakpoint of $v$, for all $b$.
\end{definition}

\begin{definition}
Let $s$ be a $b$-breakpoint of $v$. The \emph{neighborhood} of $s$ consists of the bundles $s, s+b,s+2b,\ldots,s+n\cdot b$.
\end{definition}

\begin{lemma}\label{lemma-optimize}
Let $\weight_\epsilon(\vec o)\cdot \vec o$ be a weighted allocation that maximizes the welfare in a structured range among all allocations with $\Sigma_is_i\leq s$, for some $m\geq s\geq 0$. Let $t$ be the maximal non-negative integer s.t. each $o_i$ is a multiple of $2^t$. Then, each $o_i$ is in the neighborhood of a $2^t$-breakpoint.
\end{lemma}
\begin{proof}
Suppose towards a contradiction that there is some bidder $k$, where $o_k$ is not in the neighborhood of a $2^t$-breakpoint. Round down each $o_i$ (that is not a $2^t$-breakpoint) to the nearest $2^t$-breakpoint of $v_i$ denoted $\overline o_i$. Let $c_i$ be such that $\overline o_i=c_i\cdot 2^t$. For each $i$, define $c'_i$ as follows:
\[
c'_i=\left\{
       \begin{array}{ll}
         c_i+1, & \hbox{$c_i$ is odd;} \\
         c_i, & \hbox{$c_i$ is even.}
       \end{array}
     \right.
\]
For each $i$, let $a_i=c'_i\cdot 2^t$. We will show that the expected welfare of $\weight_\epsilon(\vec a)\cdot \vec o$ is greater than that of the optimal solution $\weight_\epsilon(\vec o)\cdot \vec o$. Note that $\Sigma_ia_i\leq s$: on one hand by our assumption, $o_k-\overline o_k\geq n\cdot 2^t$. On the other hand, $\Sigma_i(a_i-\overline o_i)\leq n\cdot 2^t$, since for each $i$, $a_i-\overline o_i= 2^t$ if $\overline o_i$ is odd, else $a_i-\overline o_i=0$. (In a sense, we ``take'' items from bidder $k$ and spread them among the other bidders.)

We now calculate the expected welfare of $\vec a$. All $c'_i$'s are even, thus each $a_i$ is a multiple of $2^{t+1}$. In addition, the number of bidders that receive no items is larger in $\vec a$. Hence $\weight_\epsilon(\vec a)\geq (1+2\delta)\weight_\epsilon(\vec o)$, by the definition of $\weight_\epsilon$. Since for each $i$, $a_i\geq \overline o_i$, the monotonicity of the valuations implies that:
\begin{align*}
  \weight_\epsilon(\vec a)\cdot \Sigma_iv_i(a_i) &\geq
  (1+2\delta)\weight_\epsilon(\vec o)\cdot \Sigma_iv_i(a_i) \\
  &\geq (1+2\delta)\weight_\epsilon(\vec o)\cdot\Sigma_iv_i(\overline
  o_i)
\end{align*}
Using the definition of a breakpoint, the value of each $v_i(o_i)$ is close to the value of $v_i(\overline o_i)$:
\begin{align*}
  (1+2\delta)\weight_\epsilon(\vec o)\cdot\Sigma_iv_i(\overline o_i)
  &\geq (1+2\delta)\weight_\epsilon(\vec o)\cdot\Sigma_i\frac
  {v_i(o_i)}{1+\delta}\\
  &> \weight_\epsilon(\vec o)\cdot\Sigma_iv_i(o_i)
\end{align*}
I.e., the expected welfare of the distribution $\weight_\epsilon(\vec a)\cdot \vec a$ is greater than that of the optimal one: $\weight_\epsilon(\vec o)\cdot \vec o$. A contradiction.
\end{proof}

\subsection{Warmup: An FPTAS for a Fixed Number of Bidders}

As a warm-up we provide an MIDR $(1+\epsilon)$-approximation algorithm
for a fixed number of bidders (Figure \ref{algwarmup}). The idea is to fully optimize over a
structured range. Specifically, the number of value queries the
algorithm makes is $\poly\lt(n,\log m,\frac 1 \delta,\log \max_i
\frac{v_i(m)} {v_i(1)} \rt)=\poly\lt(n,\log m,\frac 1 \epsilon,\log
\max_i\frac{v_i(m)} {v_i(1)} \rt)$. This algorithm has two
drawbacks. The first major one is that while the number of value
queries is small, the \emph{running time} of the algorithm is
exponential in $n$: exhaustive search is used to find the best
allocation that consists only of bundles in the neighborhood of
breakpoints. However, for a fixed number of bidders the exhaustive
search (and hence the algorithm) is also computationally
efficient\footnote{Fixing the number of bidders is still an
  interesting case, as the communication lower bound \cite{NS06} and
  the lower bound for MIR algorithms \cite{DN07b} use only $2$
  bidders. Furthermore, interestingly enough an FPTAS for a constant
  number of bidders implies a computationally efficient PTAS for a
  general numbers of bidders, as a straightforward extension of the
  result of \cite{DN07b}.}.

The second drawback is that the running time and the number of value
queries depends (polynomially) on $\log \max_i\frac{v_i(m)} {v_i(1)}$
(in a sense, the algorithm is only ``weakly polynomial''). This can be
fixed (using ``significant breakpoints'' -- see the next section), but
to keep the presentation of this warmup simple we do not address this
here. We stress that the FPTAS for a general number of bidders we
present in the next subsection is ``strongly polynomial''. I.e., runs
in time $\poly\lt(n,\log m,\frac 1\epsilon \rt)$.

\begin{figure}[center,h]
\hrulefill
\begin{enumerate}
    \item For each $t=0,\ldots,\floor {\log m }$:
    \begin{itemize}
        \item Find the $2^t$-breakpoints of each bidder $i$.
        \item Let $\vec {s^t}=(s^t_1,\ldots ,s^t_{n})$ be the allocation maximizing welfare among all allocations $(s_1,\ldots s_n)$ where for each $i$, $s_i$ is in the neighborhood of a $2^t$-breakpoint of $v_i$.
    \end{itemize}
    \item Let $\weight_\epsilon(\vec s)\cdot \vec s$ be the welfare
      maximizing distribution in
      $\{\weight_\epsilon(\vec{s^0})\cdot \vec {s^0}, \ldots,
      \weight_\epsilon(\vec{s^{\floor {\log m}}}) \cdot \vec{s^{\floor
          {\log m}}} \}$. Return $\weight_\epsilon(\vec s)\cdot \vec s$.
\end{enumerate}
 \hrulefill
\caption{FPTAS for Fixed Number of Bidders}
\label{algwarmup}
\end{figure}

\begin{lemma}
  The number of value queries the algorithm makes is $\poly\lt(\log m,
  n,\frac 1 \epsilon,\log \max_i\frac{v_i(m)} {v_i(1)}\rt)$.
\end{lemma}
\begin{proof}
  For each one of the possible $\floor{\log m}$ values of $t$, we find
  all $O(\log_{1+\delta} \max_i\frac{v_i(m)} {v_i(1)})$ breakpoints of
  each bidder $i$, and then query the $n$ bundles in the neighborhood
  of the breakpoints. Binary search finds a single breakpoint in
  $O(\log m)$ queries, thus the total number of queries is $O(n^2
  \cdot \log^2 m \log_{1+\delta} \max_i\frac{v_i(m)} {v_i(1)})=\poly
  \lt(\log m, n,\frac 1 \epsilon,\log \max_i\frac{v_i(m)} {v_i(1)}
  \rt)$.
\end{proof}

The following is an immediate corollary of Lemma \ref{lemma-optimize} (setting $s=m$ in the statement of the lemma):
\begin{lemma}
The algorithm finds a welfare maximizing distribution in an $\epsilon$-structured range.
\end{lemma}

Using Lemma \ref{lemma-approx} we now have:
\begin{theorem}
There is a $(1-\epsilon)$-approximation algorithm for multi-unit
auctions that is truthful in expectation and makes $\poly\lt(n,\log m,\frac 1 \epsilon,\log \max_i\frac{v_i(m)} {v_i(1)} \rt)$ value queries.
\end{theorem}


%% file: general.tex
\subsection{The Main Result: An FPTAS for a General Number of Bidders}
This section presents a computationally-efficient FPTAS for a general number of bidders. We use structured ranges as a basic building block.

Let $u$ and $v$ be valuations of two bidders. Informally speaking,
define a ``meta bidder'' with valuation $w$ as follows: the value
$w(s)$ is set to the optimal expected welfare of allocating at most
$s$ items in a structured range for $m$ items and the two
bidders. Note that, by lemma \ref{lemma-optimize} and a
straightforward modification of the algorithm in the previous section,
each query to $w$ can be calculated in time $\poly\lt(\log m,\frac 1
\epsilon \rt)$. Apply this construction recursively, combining two meta bidders into one ``larger'' meta bidder (see Figure \ref{fig:tree}). We end up with two meta-bidders at the top level, a case solved in the previous section.  We will show that this algorithm runs in time $\poly\lt(n,\log m,\frac 1 \epsilon\rt)$, and that the solution obtained well-approximates the optimal welfare. We will also show that this algorithm optimizes over some range of distributions, yielding truthfulness in expectation. Throughout this section we assume, without loss of generality, that $n$ is a power of $2$ (as we can always add bidders that have a value of $0$ for each bundle).

\subsubsection{The Range}

Before defining the range of the algorithm, we define a binary hierarchal division of the bidders, illustrated in Figure \ref{fig:tree}. We think of each node in the hierarchy as a meta-bidder representing a subset of the bidders. The root meta-bidder, representing $\set{1,\ldots,n}$, is composed of two meta-bidders representing $\set{1,\ldots,\frac{n}{2}}$ and $\set{\frac{n}{2}+1, \ldots, n}$, respectively. Applying the division recursively, meta-bidder $\set{a,\ldots,b}$ has two children: $\set{a,\ldots,\frac{a+b -1}{2}}$ is the left child and $\set{\frac{a+b-1}{2} +1, \ldots, b}$ is the right child. The leaves correspond to single bidders. Number the tree levels from $1$ to $\log n$, where the leaves are in level $\log n$, and the top vertex is in level $1$. The following definitions will prove helpful in defining the range of the algorithm $\R$.
\begin{definition}\label{def:composition}
Let $A$ and $B$ be two disjoint set of bidders. Let $D_A$ and $D_B$ be two distributions over allocations to bidders in $A$ and $B$, respectively. Let $w\in[0,1]$. A \emph{weighted composition} $w\cdot (D_A\oplus D_B)$ is the following distribution over allocations to bidders in $A\cup B$: with probability $w$, allocate to bidders in $A$ according to $D_A$ and to bidders in $B$ according to $D_B$, and with probability $(1-w)$ allocate no items to all bidders in $A\cup B$.
\end{definition}

\begin{definition}\label{def:components}
The set of \emph{components} of a weighted composition distribution $w\cdot (D_A\oplus D_B)$ includes $D_A$ and $D_B$. If $D_A$ and $D_B$ are weighted composition distributions then the components of $D_A$ and $D_B$ are also contained in the set of components of $w\cdot (D_A\oplus D_B)$.
\end{definition}

Let $\epsilon'=\frac {\ln \frac {1} {1-\epsilon}} {\log n}$. In the remainder of this section, every reference to a structured range refers to an $\epsilon'$-structured range on $m$ items and two bidders, though we often refer to a subset of the structured range that allocates at most $k$ items, for some $k\leq m$. For each vertex $T$ in the hierarchical tree with two leaf children $l$ and $r$, let $\mathcal R^k_{T}$ be the set of all weighted allocations of at most $k$ items in a $\epsilon'$-structured range of bidders $l$ and $r$. For a vertex $T$ with two non-leaf children $L$ and $R$, let
\begin{align*}
\mathcal R^k_{T}=\{&\weight_{\epsilon'}(a,b)\cdot (D_L\oplus D_R) 
\ | \ a+b\leq k, D_L\in \mathcal R^{a}_L, D_R\in \mathcal R^{b}_R\}
\end{align*}
Now let the range of the algorithm be $\mathcal R=\mathcal R^m_{N}$.

Consider some distribution $D\in \mathcal R_U^t$ (for some $U$ and $t$). Let $T$ be a descendant of $U$ in the hierarchical tree. There is a unique $k\leq t$ such that some $D_T^k \in \R^k_T$ is a component of $D$. The distribution $D_T^k$ is said to be \emph{induced} by $D$ on $T$, and $k$ is said to be the number of items \emph{associated} with $T$ in $D$. Let $s_i$ be the number of items associated with bidder $i$ in $D$. The (feasible) allocation $(s_1,\ldots,s_n)$ is called the \emph{base allocation} of $D$. Notice that for every allocation $\vec s$ there is at least one allocation $D\in\mathcal R$ such that $\vec s$ is the base allocation of $D$. Moreover, distribution $D$ always allocates either $s_i$ items or $0$ items to bidder $i$.

\begin{figure}[center,h]
\centerline{\includegraphics[scale=0.25]{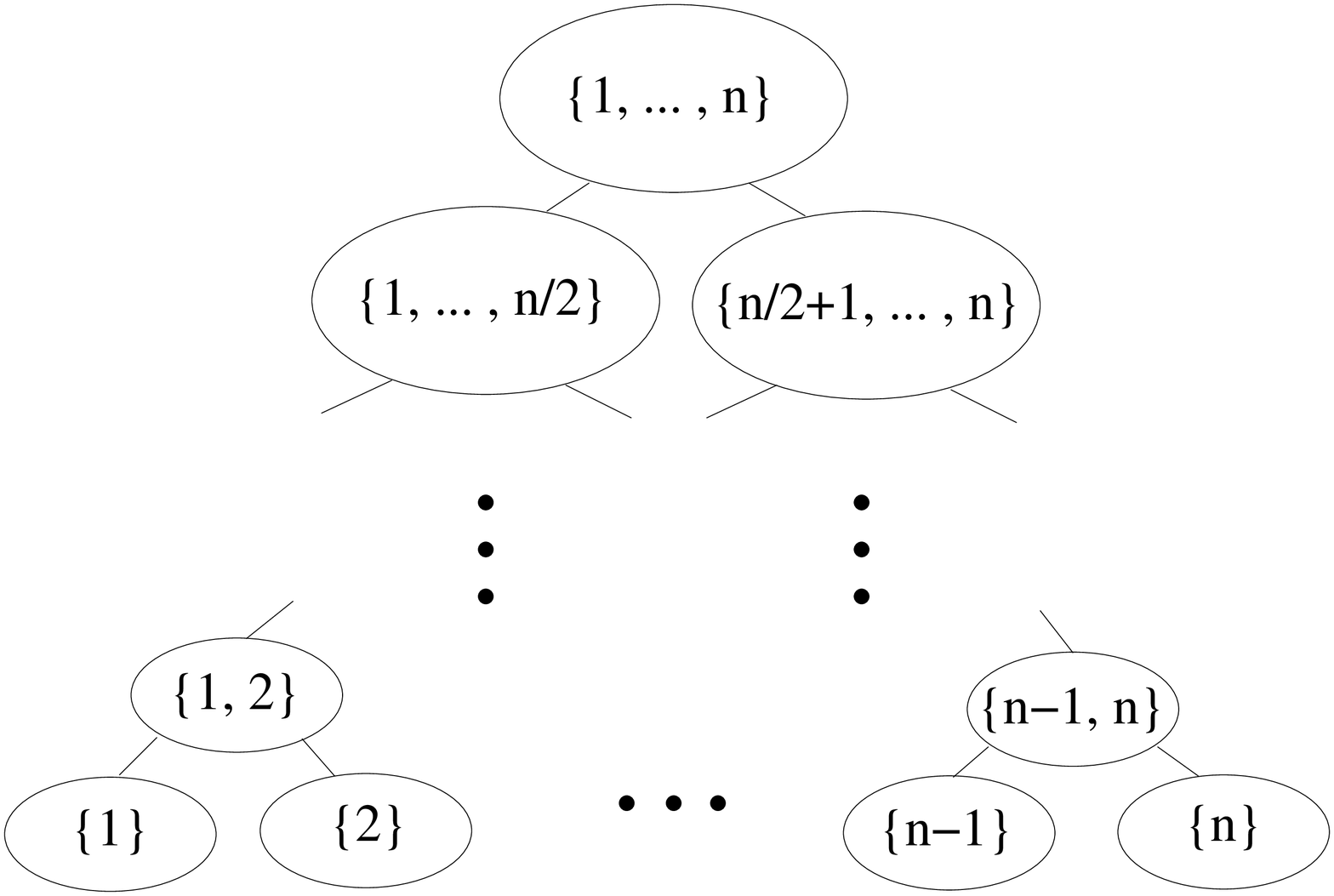}}
  \caption{Meta-Bidder Hierarchy}
  \label{fig:tree}
\end{figure}

An algorithm that optimizes over $\mathcal R$ is truthful in expectation. It also provides a good approximation:
\begin{lemma}
An algorithm that optimizes over $\mathcal R$ provides an approximation ratio of at least $1-\epsilon$.
\end{lemma}
\begin{proof}
The lemma follows from the following claim:
\begin{claim}
Let $D\in \mathcal R$. Let $\vec s$ be its base allocation. $\Pr_D[\hbox{bidder $i$ receives $s_i$}]\geq 1-\epsilon$.
\end{claim}
\begin{proof}
If $s_i= 0$ the claim is true. Otherwise, for a non-leaf vertex $T$ with children $L$ and $R$, let $w_T\cdot (D_L\oplus D_R)$ denote the weighted composition distribution induced by $D$ on $T$. $w_T$ is the probability that the child vertices of $T$ receive a non-empty bundle. By definition, $w_T \geq 1-\epsilon'$. The probability that a bidder $i$ receives $s_i$ is the probability that all non-leaf vertices containing $i$ ``allocate'' items to their children. Let $P$ be the set of all non-leaf vertices on the path in the hierarchal tree from the leaf $\set{i}$ to the root $N$. We can bound this probability as follows: \[\prod_{T\in P}w_T\geq (1-\epsilon')^{\log n}=\lt(1-\frac {\ln \frac {1} {1-\epsilon}} {\log n}\rt)^{\log n} \geq 1-\epsilon\]
\end{proof}

Observe that the optimal (unweighted) allocation $(o_1,\ldots,o_n)$ is induced by some $D\in \mathcal R$. Hence, each bidder $i$ receives $o_i$ with probability at least $(1-\epsilon)$ in $D$. The expected welfare of $D$ is therefore at least $(1-\epsilon)$ times the welfare of the optimal unweighted allocation. The algorithm optimizes over all distributions in the the range $\R$, and thus returns an allocation with at least that expected welfare.
\end{proof}

We now make the meta-bidder intuition formal.

\begin{definition}
Let $\mathcal D$ be a range of distributions of $m$ items to a set $U$ of bidders. Let $\mathcal D^k$ be the set of all distributions in $\mathcal D$ that never allocate more than $k$ items. Let $S^U_k$ be the set of all allocations that allocate at most $k$ items to bidders in $U$. Define $v$ to be the valuation of a \emph{meta-bidder} over $U$ using $D$ as follows:
$$v(k)=\max_{D\in\mathcal D^k}\sum_{\vec s\in S^U_k}\Pr_D[\vec s \hbox { is chosen}]\cdot \sum_{i\in U} v_i(s_i)$$
\end{definition}
Note that a valuation of a meta bidder is monotone and normalized. We now define a meta-bidder for each vertex $T$. The definition is recursive. For a vertex $T$ with two leaf children $l$ and $r$, let $v_T$ be the valuation of the meta bidder $T$ defined over $\set{l,r}$ using the range of distributions $\mathcal D=\cup_k\mathcal R^k_T$ (i.e., an $\epsilon'$-structured range). For a vertex $T$ with two meta-bidder children $L$ and $R$, let $v_T$ be the valuation of the meta-bidder $T$ defined over $\set{L,R}$ using the the $\epsilon'$-structured range of allocations to (meta-) bidders $L$ and $R$. The meta-bidder idea is powerful for precisely the following reason:

\begin{lemma}\label{lemma-meta-bidders}
Let $O$ be a welfare maximizing distribution in $\mathcal R^s_T$. Let $O_U$ be the distribution of allocations to bidders in $U$ induced by $O$ on $U$, where $U$ is a descendant of $T$ in the hierarchical tree or $U=T$. Denote the expected welfare of $O_U$ by $|U|$. Let $o_U$ be the bundle that is associated with $U$ in $O$. Then, $v_U(o_U)=|U|$.
\end{lemma}
\begin{proof}
The proof is by induction on the level of $U$. For a vertex $U$ with two leaf children $l$ and $r$, $v_U(o_U)$ equals the value of the welfare maximizing solution that allocates at most $o_U$ items in a structured range. On the other hand, observe that $O$ induces a weighted allocation that allocates at most $o_U$ items in a structured range, and that $O$ is welfare maximizing. Hence, the expected welfare of $O_U$ equals $v_U(o_U)$.

Consider a vertex $U$ with two leaf children $L$ and $R$, and let
$O_U=\weight_{\epsilon'}(a,b)\cdot (O_L\oplus O_R)\in \mathcal
R_U^{o_U}$. The expected welfare of $O_U$ is by induction
$\weight_{\epsilon'}(l,r)\cdot (v_L(a) + v_R(b))$, i.e., $O_U$ is the welfare maximizing weighted allocation in a structured range with bidders $L$ and $R$ that allocates at most $o_U$ items, which is by definition the value of $v_U(o_U)$.
\end{proof}

Fix $\delta=\frac {\ln \frac 1 {1-\epsilon'}} {2\log m+2}$. The important property of distributions in the range is that, similar to structured ranges, the bundle $k$ that is associated with a vertex is in the neighborhood of a breakpoint of the corresponding (meta-) bidder.

\begin{lemma}\label{lemma-step-point}
Let $O$ be a welfare maximizing distribution in $\mathcal R^a_T$. Let $L$ and $R$ be two sibling vertices that are descendants of $O$. Let $o_L$ and $o_R$ be the bundles associated with $L$ and $R$ in $O$, respectively. Let $t$ be the maximal non-negative integer such that $o_L$ and $o_R$ are multiples of $2^t$. Then, $o_L$ and $o_R$ are in the neighborhood of a $2^t$-breakpoint of $v_L$ and $v_R$, respectively.
\end{lemma}
\begin{proof}
Let $P$ be the parent of $L$ and $R$, let its associated bundle in $O$ be $o_P$. By Lemma \ref{lemma-meta-bidders} and the definition of meta bidders, $v_P(o_P)$ is equal to the expected welfare of a distribution of the form $\weight_{\epsilon'}((o_L,o_R))\cdot (o_L,o_R)$ -- the welfare maximizing distribution that allocates at most $o_P$ items to bidders $L$ and $R$ in a structured range. By Lemma \ref{lemma-optimize}, $o_L$ and $o_R$ are in neighborhoods of $2^t$-breakpoints of $v_L$ and $v_R$, respectively.
\end{proof}

\subsubsection{The FPTAS}

The FPTAS we construct is the obvious one: find the optimal allocation to meta-bidders $\set {1,\ldots n/2}$ and $\set {n/2 +1,n}$ in a structured range, then proceed recursively to find the best allocation between bidders $\set{1,\ldots,\frac n 4}$ and $\set{\frac n 4+1,\ldots, \frac n 2}$, and so on. In a naive implementation of this algorithm every value query to the valuations of the two meta-bidders is calculated recursively ``on the fly''. However, a short calculation shows that in this implementation the number of value queries is polynomial in $(\log m)^{\log n},n,\frac 1 \epsilon$. To improve the running time, we use the fact that to calculate $v_T(s)$, for any bundle $s$, we only need to know the value of a relatively small number of bundles in the neighborhood of breakpoints of the two (meta-) bidders that are $T$'s children.

We are also interested in a ``strongly polynomial'' algorithm. I.e., the running time should be independent of $\log \frac {\max_iv_i(m)} {\min_iv_i(1)}$. Towards this end, define:
\begin{definition}
A bundle $s$ is an $r$-significant $b$-breakpoint of $v$ if it is a $b$-breakpoint of $v$ and if either $s=0$ or $v(s)\geq r$.
\end{definition}
By the next lemma, we can ``ignore'' insignificant breakpoints:
\begin{lemma}\label{lemma-significant}
Let $O$ be a welfare maximizing distribution in $\mathcal R_T^s$. Denote its expected welfare by $|O|$. Let $R,L$ be two sibling descendants of $T$ in level $l$. Let $o_L,o_R$ be the bundles associated with $L,R$ in $O$, respectively. Let $t>0$ be the largest integer s.t. both $o_L$ and $o_R$ are multiples of $2^t$. Then, $o_L$ and $o_R$ are in the neighborhood of an $r$-significant $2^t$-breakpoint of $v_L$ and $v_R$, respectively, for $r={ \delta}^l|O|$.
\end{lemma}
\begin{proof}
From Lemma \ref{lemma-step-point} we know that $o_L$ and $o_R$ are in the neighborhood of $2^t$-breakpoints of $v_L$ and $v_R$. It remains to show that they are $r$-significant. We make use of the following claim:
\begin{claim}
Let $\weight_\epsilon(\vec o)\cdot \vec o$ be a welfare maximizing distribution of $s$ items in an $\epsilon$-structured range for $2$ bidders. Denote its expected welfare by $|o|$. For each $i\in T$ with $o_i\neq 0$ we have that $v_i(o_i)\geq \delta |o|$.
\end{claim}
\begin{proof}
The proof is almost identical to the proof of Lemma \ref{lemma-optimize}, and therefore it is only sketched. Suppose for contradiction that there is some bidder $k$ with $v_k(o_k)< \delta |o|$ and $o_k>0$. Let $t>0$ be the maximal integer s.t. each $o_i$ is a multiple of $2^t$. Define the following allocation $\vec a$: $a_k=0$, and for the other bidder $i\neq k$, let $a_i=o_i$ (notice that $\Sigma_ia_i\leq s$). All bidders but one receive no items, thus $\weight_{\epsilon}(\vec a)\geq (1+2\delta)\weight_{\epsilon}(\vec o)$, by definition of $\weight_{\epsilon}$. We have that $\Sigma_iv_i(a_i)\geq (1-\delta)\Sigma_iv_i(o_i)$. However, now $\weight_\epsilon(\vec a)\Sigma_iv_i(a_i)>\weight_\epsilon(\vec o)\Sigma_iv_i(o_i)$, a contradiction.
\end{proof}

We prove the lemma using the claim by induction on the level of the tree, starting with children of $T$ to the leaves. By Lemma \ref{lemma-meta-bidders}, $v_T(s)=|O|$. $v_T$ is a meta bidder, thus $v_T(s)$ is also equal to the value of a welfare maximizing allocation in a structured range: $w\cdot (o_L,o_R)$. By the claim $v_L(o_L)\geq \delta  |O|$. Similarly, $v_R(o_R)\geq \delta |O|$.

Assume correctness for level $l$ and prove for $l+1$. For $L$ and $R$ in level $l+1$ with parent $P$, from arguments similar to those above it follows that $v_L(o_L),v_R(o_R)\geq \frac \delta 2 v_P(o_P)$, where $o_L,o_R$, and $o_P$ are the bundles associated with $L,R$ and $P$ in $O$, respectively. By induction we get that: $v_R(o_R),v_L(o_L)\geq \delta  v_p(o_P)\geq {\delta }^{l+1}v_T(s)$, as needed.
\end{proof}

\begin{figure}[h]
\hrulefill
\begin{enumerate}
    \item For each vertex $T$ in level $l$ evaluate (recursively, from bottom to top) all bundles that are in the neighborhood of an $r_l$-significant $2^t$-breakpoint of $v_T$ (for $t=1,2\ldots,\floor{\log m}$, and $r_l=\delta^l\max_{i\in N}v_i(m)/2$)
    \item Evaluate (recursively, from top to bottom starting from the top vertex $N$ and $o_N=m$) for each vertex $T$ with children $L$ and $R$ a welfare maximizing weighted distribution $\weight_{\epsilon'}((o_L,o_R))\cdot (o_L,o_R)$ of (meta-) bidders $L$ and $R$, where $o_L+o_R\leq o_T$. Let $w_T=\weight_{\epsilon'}((o_L,o_R))$.
    \item Each bidder $i$ receives $o_i$ items with probability $\Pi_{T\in P} w_T$ where $P$ is the set that includes all the proper ancestors of $i$ in the hierarchal tree.
\end{enumerate}
\hrulefill
\caption{An FPTAS for a General Number of Bidders}
\label{alggeneral}
\end{figure}

\begin{theorem}\label{thm-general-fptas}
The algorithm in  Figure \ref{alggeneral} is a truthful-in-expectation FPTAS for multi-unit auctions.
\end{theorem}
Before proving the theorem, we comment on the definition of $r_l$ in the algorithm. Let $|O|$ be the expected welfare of $O$. On one hand, observe that $|O|\geq \frac {\max_i v_i(m)} 2$ (allocating all items to bidder $i$ that maximizes $v_i(m)$ is an allocation that is induced by some $D\in \mathcal R$. $D$ has an expected welfare of at least $(1-\epsilon)v_i(m)$), and on the other hand $|O|\leq n\max_i v_i(m)$. In other words, $r_l$ is on one hand small enough so we can calculate the significant breakpoints of the valuations according to Lemma \ref{lemma-significant}, and on the other hand is not too small, so we do not have too many significant points to evaluate. The proofs below make this discussion formal.

\begin{lemma}
The algorithm finds a welfare maximizing distribution $O\in \mathcal R$.
\end{lemma}
\begin{proof}
We prove the following claim first:
\begin{claim}
Let $T$ be a vertex in level $l$. The $\delta ^l|O|$-significant $2^t$-breakpoint are evaluated correctly (for $t=1,2,...,\floor{\log n}$).
\end{claim}
\begin{proof}
The proof is by induction on the level $l$, starting from the leaves to the top. For $T$ with two leaf children $L$ and $R$, the claim is trivially true, since $|O|\geq \frac {\max_i v_i(m)} 2$ and by Lemma \ref{lemma-significant}.

We assume correctness for $l+1$ and prove for $l$. Consider vertex $T$ in level $l$ with meta-bidder children $L$ and $R$. By induction the values of bundles in the neighborhood of $r_{l+1}$-significant breakpoints of $v_L$ and $v_R$ are computed correctly. Consider some bundle $a_T$ in the neighborhood of an $r_l$-significant breakpoint of $v_T$. Denote by $\weight_{\epsilon'}((a_L,a_R))\cdot (a_L,a_R)$ the welfare maximizing weighted allocation to $v_L$ and $v_R$ that allocates at most $a_T$ items. By induction, the values of $v_L(a_L) ,v_R(a_R)$ are calculated correctly, since they are $\delta^{l+1}$-significant breakpoints of $v_L$ and $v_R$ (Lemma \ref{lemma-significant}). Hence $v_T(a_T)$ is calculated correctly, as needed.
\end{proof}

Now we can prove the lemma by induction from the top vertex to the leaves. The welfare of the distribution induced by $O$ on the top vertex $N$ is, by Lemma \ref{lemma-meta-bidders}, equal to a welfare maximizing weighted allocation to bidders $\set{1,\ldots ,\frac n 2}$ and $\set{\frac n 2+1,\ldots,n}$ in a structured range $\weight_{\epsilon}\lt(\lt(o_{\set{1,\ldots ,\frac n 2}},o_{\set{\frac n 2+1,\ldots,n}}\rt)\rt)\cdot \lt(o_{\set{1,\ldots ,\frac n 2}},o_{\set{\frac n 2+1,\ldots,n}}\rt)$. By the same arguments as above, the algorithm finds a welfare maximizing weighted allocation to bidders $\set{1,\ldots, \frac n 4}$ and $\set{\frac n 4+1,\ldots, \frac n 2}$ of at most $o_{\set{1,\ldots ,\frac n 2}}$ items. Proceeding similarly until we reach the leaves, we end up with a distribution in $\mathcal R$ with an expected welfare equal to that of $O$.
\end{proof}

\begin{lemma}
The algorithm runs in time $\poly\lt(n,\log m,\frac 1 \epsilon \rt)$.
\end{lemma}
\begin{proof}
Fix some vertex $T$ with children $R$ and $L$. Observe that $v_T(s)$, for any $s$, can be computed using only bundles in the neighborhood of $r$-significant $b$-breakpoints of $v_L$ and $v_R$ (for $b=1,2,4,\ldots,m$).

To see that these values can be obtained efficiently, notice that each valuation $v$ in level $l$ has only $\poly\lt(n,\frac 1 \delta\rt)$ $r_l$-significant $b$-breakpoints (by our choice of $r$ -- since $2n\cdot max_iv_i(m)\geq |O|$). Recall that via binary search a $b$-breakpoint of $v$ can be can be found in $\poly(\log m)$ value queries to $v$. For each breakpoint we also need to query all bundles in its neighborhood. Hence we need to query each valuation $v$ (of a bidder or a meta bidder) only $\poly\lt(\log m,\frac 1 \delta\rt)$ times. The number of valuations of bidders and meta-bidders in the tree is $O(n)$. and thus the total number of queries is still $\poly\lt(n,\log m,\frac 1 \delta\rt)$. Notice that the computational overhead is polynomial in the number of value queries: for each value query we compute a welfare maximizing allocation of the previous level in a structured range of $2$ bidders, which can be done in time $\poly\lt(\log m,\frac 1 \delta\rt)$. In total, the running time of the algorithm is $\poly\lt(n,\log m,\frac 1 \delta\rt)=\poly\lt(n,\log m,\frac 1 \epsilon\rt)$.
\end{proof}

\noindent This completes the proof of Theorem \ref{thm-general-fptas}


%% file: separation.tex
\section{Truthful in Expectation Mechanisms Have More Power}\label{sec-separation}

This section shows that polynomial-time truthful-in-expectation mechanisms are strictly more powerful than universally-truthful polynomial-time mechanisms. Ideally, we would like to prove this for multi-unit auctions. However, finding a non-trivial lower bound on the power of polynomial time mechanisms for multi-unit auctions, even deterministic ones, is a big open question. Hence we study a variant of multi-unit auctions that is artificially restricted, with the sole purpose of proving such a separation for the first time. We prove that universally truthful polynomial time mechanisms cannot provide an approximation ratio better than $2$ for this variant, while we provide a truthful in expectation FPTAS. As in multi-unit auctions, the (deterministic) VCG mechanism solves this problem optimally, but in exponential time.

In a \emph{restricted multi-unit auction}, a set of $m$ items, where $m$ is a power of $2$, is to be allocated to two bidders. The set of feasible allocations is restricted as follows: either no items are allocated, or all items are allocated with at least one item per bidder. Each bidder $i\in \{1,2\}$ has a valuation function $v_i$, given as a black box, specifying the bidder's value for each number of items. We restrict $v_i$ to be normalized ($v_i(0)=0$) and strictly increasing. The objective is to maximize social welfare, the sum of the values of the bidders. Our algorithms should run in time polynomial in $\log m$.

\begin{theorem}\label{thm-separation}
The following statements are true:

\begin{enumerate}
\item For every constant $\epsilon>0$, every universally truthful $(2-\epsilon)$-approximation mechanism for restricted multi-unit auction requires $\Omega(m)$ communication.

\item There exists a $(1+\epsilon)$-approximation algorithm for restricted multi-unit auctions that is truthful in expectation and runs in time $\poly\lt(\log m,\frac 1 \epsilon\rt)$ (an FPTAS).
\end{enumerate}
\end{theorem}

Before proceeding with the proof, we discuss the differences between
restricted multi-unit auctions and the standard multi-unit auctions
discussed in previous sections, and the role these differences play in
the proof.

\begin{itemize}
\item \textbf{Strictly Increasing Valuations, $m$ is a power of $2$: } These restrictions are only there to simplify the proof. They can be removed without changing the statement of the theorem.

\item \textbf{Allocate all Items with at least One Item per Bidder or
    Allocate Nothing: } The all-items-are-allocated constraint fulfills the conditions of the characterization of deterministic truthful mechanisms for multi-unit auctions \cite{DS08-add}. The restriction that each bidder must receive at least one item, and the relaxation allowing the empty allocation, will prove useful in arguing about universally truthful randomized mechanisms which were not considered in \cite{DS08-add}.

\end{itemize}

To prove Theorem \ref{thm-separation}, first we bound the power of
universally truthful polynomial time mechanisms for restricted
multi-unit auctions. Then, we show that there exists a truthful in
expectation FPTAS for restricted multi-unit auctions.

\subsection{A Lower Bound on Universally Truthful Mechanisms}

\subsubsection{From Universally Truthful to Deterministic Mechanisms}

It is inconvenient to study randomized mechanisms directly. Therefore, we start by showing that the existence of a universally truthful mechanism with a good approximation ratio implies the existence of a deterministic mechanism that provides a good approximation on ``many'' instances. The following definition and propositions are adapted from \cite{DN07a}. We repeat the proof here for completeness.

\begin{definition}
Fix $\alpha \geq 1$, $\beta\in [0,1]$, and a finite set $U$ of instances of restricted multi-unit auctions. A deterministic algorithm $B$ for restricted multi-unit auctions is \emph{$(\alpha,\beta)$-good} on $U$ if $B$ returns an $\alpha$-approximate solution for at least a $\beta$-fraction of the instances in $U$.
\end{definition}

\begin{proposition}[essentially \cite{DN07a}]\label{prop-good-mechanism}
Let $U$ be some finite set of instances of restricted multi-unit auctions, and let $\alpha,\gamma\geq 1$. Let $A$ be a universally truthful mechanism for restricted multi-unit auctions that provides an expected welfare of $\frac{OPT(I)} \alpha$ for every instance $I \in U$ with expected communication complexity $cc(A)$. Then, there is a $\gamma\cdot cc(A)$-time (deterministic) algorithm in the support of $A$ that is $\lt(\frac 1 {\frac 1 \alpha -\frac 1 \gamma},{\frac 1 \alpha -\frac 1 \gamma}\rt)$-good on $U$ and has a communication complexity of $\gamma\cdot cc(A)$.
\end{proposition}
\begin{proof}
For each mechanism $D$ in the support of $A$, Let $cc(D)$ denote its expected communication complexity. Let $A'$ be the mechanism obtained from $A$ by replacing in the support of $A$ each mechanism $D$ for which $cc(D)\geq  \gamma cc(A)$ by a mechanism that never allocates any items. Notice that the expected approximation ratio of $A'$ is at least $\alpha'=\frac 1 {\frac 1 \alpha -\frac 1 \gamma}$.

Fix some instance $I\in U$. The \emph{expected} welfare of $A'(I)$ is at least $\frac {OPT(I)} {\alpha'}$. The \emph{probability} that $A(I)$ has welfare at least $\frac {OPT(I)} {\alpha'}$ is lower-bounded by $\frac 1 {\alpha'}$: this is achieved if $A(I)$ is optimal with probability $\frac 1 {\alpha'}$ and has welfare $0$ otherwise. Hence, there exists a deterministic algorithm $B$ in the support of $A$ that returns an $\alpha'$-approximate solution for at least a $\frac 1 {\alpha'}$-fraction of the instances in $U$. Notice that the communication complexity of $B$ is $\gamma\cdot cc(A)$.
\end{proof}

Let $A$ be a universally truthful mechanism for restricted multi-unit auctions with communication complexity $\poly(\log m)$ and an approximation ratio of $2-\epsilon$, for a constant $\epsilon>0$. By the claim, there must exist a $\lt(2-\epsilon',\frac 1 2\rt)$-good deterministic mechanism in its support with communication complexity of $\log m\cdot cc(A)$ (for some constant $\epsilon'>0$ and some set of instances $U$). We will show that the communication complexity of this ``good'' mechanism for that $U$ is $\Omega(m)$, thus the communication complexity of $A$ is $\Omega\lt(\frac m {\log m}\rt)$.

\subsubsection{A Lower Bound on ``Good'' Deterministic Mechanisms}

To start, we recall the following characterization result for deterministic mechanisms for multi-unit auctions:

\begin{theorem}[\cite{DS08-add}]\label{thm-char}
Let $f$ be a truthful deterministic mechanism for multi-unit auctions with strictly monotone valuations that always allocates all items with range at least $3$. Then, $f$ is an affine maximizer.
\end{theorem}

%

Fix some universally truthful mechanism $A$ for restricted multi-unit auctions. The support of $A$ contains three possible (non-disjoint) types of deterministic mechanisms:
\begin{enumerate}
\item \textbf{Imperfect Mechanisms:} Mechanisms $B$ where there exist
  valuations $u,v$ such $B(u,v)=(0,0)$.
\item \textbf{Tiny-Range Mechanisms:} Mechanisms that have a range of size at most $2$.
\item \textbf{Affine Maximizers:} Mechanisms that are affine maximizers.
\end{enumerate}

Notice that there are no other mechanisms in the support of $A$: a mechanism for restricted multi-unit auctions that always allocates some items must always allocate all items. If its range is of size at least $3$, then by Theorem \ref{thm-char} it must be an affine maximizer.

\begin{claim}
Let $B$ be a deterministic imperfect mechanism for restricted multi-unit auctions. Then, there is a constant $C_B>0$ such that $B(u',v')=(0,0)$ whenever $u'(m),v'(m)< C_B$.
\end{claim}
\begin{proof}
Let $u,v$ be such that $B(u,v)=(0,0)$. Let $C_B=\min_{k>0}\min(u(k),v(k))$. Since valuations are normalized and strictly increasing, $C_B >0$. Consider some $u',v'$ as in the statement of the claim. Observe that $B(u,v')=(0,0)$: by weak monotonicity (see \cite{Nis07}) bidder $1$ should be allocated no items, hence bidder $2$ is allocated no items (recall that every non-empty feasible allocation assigns at least one item to each bidder). Similarly, $B(u',v')=(0,0)$.
\end{proof}

Let $A$ be a universally truthful mechanism for restricted multi-unit auctions that provides an approximation ratio of $2-\epsilon$.  Define $C_A=\min_BC_B$, where the minimum is taken over all imperfect mechanisms in $A$'s support. Fix a constant $\sigma$ such that $\sigma << \epsilon$. For each integer $k$ such that $1 \leq k \leq m-1$, define the instance $I_k =(u_k,v_k)$ as follows: $u_k(t)= (\sigma C_A/2m) t$ for all $t<k$, and $u_k(t) = C_A/2$ for all $t\geq k$; $v_k(t)=(\sigma C_A/2m) t$ for all $t<m-k$, and $v_k(t)=C_A/2$ for all $t\geq m-k$. Let $U=\set{I_k}_k$. Notice that the optimal welfare for $I_k$ is $C_A$ and is achieved by the allocation $(k,m-k)$. Any other allocation provides a welfare of at most $C_A/2 + \sigma C_A/2 < C_A / (2-\epsilon)$.

\begin{claim}
Let $A$ be a universally truthful mechanism for restricted multi-unit auctions that provides an approximation ratio of $2-\epsilon$, for some constant $\epsilon$. Then, there exists an affine maximizer in the support of $A$ with a range of size $\Omega(m)$ that is $\lt(2-\epsilon',\frac 1 2\rt)$-good on $U$ where both bidders have positive weights, for some constant $\epsilon'>0$.
\end{claim}
\begin{proof}
By Proposition \ref{prop-good-mechanism} there exists a deterministic
mechanism $B$ in the support of $A$ that provides a
$(2-\epsilon')$-approximation on more than half of the instances in
$U$, for some constant $\epsilon'>0$ (choose $\gamma=\log m$). Notice
that $u_k(m),v_k(m) < C_A$ for each instance $I_k=(u_k,v_k)$ in $U$,
and as a result any imperfect mechanism in the support of $A$ does not
allocate any items on any instance in $U$.  Therefore, $B$ cannot be
imperfect. $B$ provides an approximation ratio of $2-\epsilon$ for
more than half of the instances in $U$, that is, outputs the optimal
allocation for more than half of the instances in $U$. Every instance
in $U$ has a different optimal allocation, and thus the range of $B$
is of size at least $m/2$. In particular, $B$ is not a tiny-range
mechanism. Thus, $B$ allocates all items and has a range of size at
least $\frac m 2$. By Theorem \ref{thm-char} it must be an affine
maximizer. Observe that the weights of both bidders must be positive:
otherwise the mechanism is a tiny range mechanism. 
\end{proof}

To finish the proof of the lower bound the following claim suffices:

\begin{claim}[essentially \cite{DN07b}]
An affine maximizer with a range of size $t$ and a positive weight for each bidder has communication complexity of at least $t$.
\end{claim}
\begin{proof}
Let $B$ be an affine maximizer. Denote its range by $\mathcal
R$. There exists non-negative constants $w_1,w_2$, and constants
$\{c_{(s_1,s_2)}\}_{(s_1,s_2) \in {\mathcal R}}$, such that
$B(u,v)=\max_{(s_1,s_2)\in \mathcal R}(w_1 u(s_1)+ w_2 v(s_2)+c_{(s_1,s_2)})$ for all valuations $u,v$.

We reduce from the disjointness problem on $t$ bits. In this problem Alice holds a string $(a_1,\ldots ,a_t)\in \{0,1\}^t$ and Bob holds a string $(b_1,\ldots ,b_t)\in \{0,1\}^t$. The goal is to determine if there exists some $i$ such that $a_i=b_i=1$. It is known that any deterministic algorithm for this problem has a communication complexity of $t$.

Let $p>>\max_{(s_1,s_2)}|c_{(s_1,s_2)}|$. Define some one-to-one and
onto correspondence $f:\mathcal R\rightarrow [t]$. Define the
following valuations: $u(i)=(2i\cdot p+a_{f((i,m-i)}\cdot p)/w_1$ if
$(i,m-i)\in \mathcal R$, otherwise $u(i)=2i\cdot p / w_1$. Also define
$v(i)=(2i\cdot p+b_{f((m-i,i))}\cdot p)/w_2$ if $(m-i,i)\in \mathcal
R$, otherwise $v(i)=2i\cdot p / w_2$. The (optimal) solution returned
by $B$ on input $(u,v)$ has welfare $(2m+2) p$ if and only if there is some $i$ with $a_i=b_i=1$. Hence, the communication complexity of $B$ is at least $t$.
\end{proof}
%
%
%
%

\subsection{An FPTAS for Restricted Multi-Unit Auctions}

All that is left is to show a better-than-2 truthful in expectation mechanism for restricted multi-unit auctions:
\begin{lemma}
There exists a truthful-in-expectation FPTAS for restricted multi-unit auctions.
\end{lemma}

 A simple variation of the FPTAS for a fixed number of
bidders yields an FPTAS for restricted multi-unit
auctions. Specifically, we will show that there is a $\poly\lt(\log
m,\frac 1 \epsilon\rt)$ algorithm for multi-unit auctions that
provides a $(1+\epsilon)$-approximation to the optimal solution in the
range of allocations permitted in restricted multi-unit
auctions. Furthermore, this algorithm optimizes over some range of
distributions over allocations $\mathcal R$. The support of each
distribution $D\in\mathcal R$ contains only allocations that permitted
in restricted multi-unit auctions. This implies the existence of a
truthful in expectation $(1+\epsilon)$ approximation algorithm for
restricted multi-unit auctions.

We start by defining a new weight function for restricted multi-unit auction:
\[
\rweight_\epsilon(\vec s)=\left\{
  \begin{array}{ll}
    0, & \hbox{$\vec s=(m,0)$ or $\vec s=(0,m)$;}\\
    \weight_\epsilon(\vec s), & \hbox{otherwise.}
  \end{array}
\right.
\]
A \emph{restricted $\epsilon$-structured weight} is defined similarly
to a structured range, but with respect to $\rweight_\epsilon$:
$\set{\rweight_{\epsilon}(\vec s)\cdot \vec s}_{\vec s\in S}$, where
$S$ is the set of all allocations of $m$ items to these $n$
bidders. The definition of a weighted allocation is also done now
using $\rweight_\epsilon$. The range of the algorithm $\mathcal R$
will be restricted $\epsilon$-structured range of allocations of $m$
items to the two bidders. The algorithm is very similar to the FPTAS
for a constant number of bidders, but with the following changes: the
expected welfare of a weighted allocation calculated in step $2$ is
defined using $\rweight_\epsilon$, and not using $\weight_\epsilon$;
in the second bullet of step $1$ we consider only allocations that are
permitted for restricted multi-unit auctions.

The correctness of the FPTAS described above follows by
straightforward modifications to the proofs of section
\ref{sec:fptas} for a constant number of bidders.
